\documentclass{article}
\usepackage{authblk,fullpage}
\usepackage{microtype}
\usepackage{mathtools,amsmath,amssymb,amsthm}
\usepackage{graphicx}
\usepackage{braket}
\usepackage{booktabs,multirow}
\usepackage[noend]{algpseudocode}
\usepackage{algorithm}
\usepackage{algorithmicx}
\usepackage{enumitem}
\usepackage{url}
\usepackage{xcolor}

\newcommand{\ignore}[1]{}

\newcommand{\symb}[1]{$\mathsf{#1}$}
\newcommand{\sym}[1]{\textsf{#1}}
\newcommand{\Puzzle}{\mathsf{puzzle}}
\newcommand{\PuzzleSol}{\mathsf{pk, nonce}}

\newcommand{\sig}[1]{\braket{#1}}

\newcommand{\puzzleTime}{D}

\newcommand{\epoch}{lifespan}
\newcommand{\epochs}{lifespans}
\newcommand{\CEV}{c,e,v}
\newcommand{\CEZ}{c,e,0}
\newcommand{\CEVS}{c,e,v,s}
\newcommand{\Members}{\mathcal{M}}
\newcommand{\TXs}{\{\mathsf{tx}\}}
\newcommand{\OracleH}{\mathcal{H}}

\newcommand{\Status}{\mathsf{status}}
\newcommand{\Propose}{\mathsf{propose}}
\newcommand{\Prepare}{\mathsf{prepare}}
\newcommand{\Commit}{\mathsf{commit}}
\newcommand{\Notify}{\mathsf{notify}}

\newcommand{\Blame}{\mathsf{view\text{-}change}}
\newcommand{\NewView}{\mathsf{new\text{-}view}}
\newcommand{\Repropose}{\mathsf{repropose}}

\newcommand{\AcceptCert}{\mathcal{A}}
\newcommand{\CommitCert}{\mathcal{C}}

\newcommand{\BlameSum}{\mathcal{V}}
\newcommand{\StatusSum}{\mathcal{S}}

\algdef{SE}[UPON]{Upon}{EndUpon}[1]{\textbf{On
}\ #1\ \algorithmicdo}{\algorithmicend\ \textbf{}}%
\algtext*{EndUpon}

\algdef{SE}[RECEIVING]{Receiving}{EndReceiving}[1]{\textbf{On receive}\ #1\ \algorithmicdo}{\algorithmicend\ \textbf{}}%
\algtext*{EndReceiving}

\newtheorem{theorem}{Theorem}

\newtheorem{lemma}{Lemma}

\begin{document}

\newcommand{\name}{Solida}

\title{\name: A Blockchain Protocol Based on \\ Reconfigurable Byzantine Consensus} 

\let\svthefootnote\thefootnote
\let\thefootnote\relax
\footnotetext{$^*$
Solidus was a gold coin used in the Byzantine Empire. \name{} is our way of (mis-)spelling Solidus.}
\footnotetext{$^{**}$
An earlier version of this paper contains incentive designs (https://arxiv.org/abs/1612.02916v1).
This version adds rigorous analysis of safety and liveness assuming an honest supermajority.
We leave rigorous analysis of incentives to future work.}
\let\thefootnote\svthefootnote

\author[1]{Ittai Abraham}
\author[2]{Dahlia Malkhi}
\author[3]{Kartik Nayak}
\author[4]{Ling Ren}
\author[5]{Alexander Spiegelman}

\affil[1]{VMware Research, Palo Alto, USA --
  \texttt{iabraham@vmware.com}}
\affil[2]{VMware Research, Palo Alto, USA -- \texttt{dmalkhi@vmware.com}}
\affil[3]{University of Maryland, College Park, USA -- \texttt{kartik@cs.umd.edu}}
\affil[4]{Massachusetts Institute of Technology, Cambridge, USA -- \texttt{renling@mit.edu}}
\affil[5]{Technion, Haifa, Israel -- \texttt{sasha.speigelman@gmail.edu}}

\date{}

\maketitle

\begin{abstract}
The decentralized cryptocurrency Bitcoin has experienced great success but also encountered many challenges. 
One of the challenges has been the long confirmation time. 
Another challenge is the lack of incentives at certain steps of the protocol, 
raising concerns for transaction withholding, selfish mining, etc. 
To address these challenges, we propose Solida, a decentralized blockchain protocol based on reconfigurable Byzantine consensus augmented by proof-of-work.
Solida improves on Bitcoin in confirmation time, and provides safety and liveness assuming the adversary control less than (roughly) one-third of the total mining power. 
\end{abstract}

\section{Introduction}
\label{sec:intro}

Bitcoin is the most successful decentralized cryptocurrency to date.
Conceptually, what a decentralized cryptocurrency needs is consensus in a \emph{permissionless} setting:
Participants should agree on the history of transactions, and anyone on the network can join or leave at any time. 
Bitcoin achieves permissionless consensus using what's now known as Nakamoto consensus. 
In Nakamoto consensus, participants accept the longest proof-of-work (PoW) chain as the history of transactions, and also contribute to the longest chain by trying to extend it.
Thus, cryptocurrencies are also known as public ledgers or blockchains in the literature. 
While enjoying great success, Bitcoin does have several drawbacks.
The most severe one is perhaps its limited throughput and long confirmation time of transactions.
For instance, presently a block can be added every ten minutes on average and it is suggested that one waits for the transaction to be six blocks deep.
This implies a confirmation time of about an hour.
Furthermore, each block can contain about 1500 transactions~\cite{Bitcoin_tx_cnt} with the current 1MB block size limit. 
This yields a throughput of  $\sim$2.5 transactions per second. 

A number of attempts were made to improve the throughput and confirmation time of Bitcoin. 
These fall into two broad categories. 
One category~\cite{GHOST,LSZ15,Bitcoin-NG} tries to improve Nakamoto consensus.
The other category~\cite{PeerCensus,ByzCoin,Elastico16,HybridConsensus} hopes to replace Nakamoto consensus with classical Byzantine fault tolerant (BFT) consensus protocols.
These proposals envision a rolling committee that approves transactions efficiently using a Byzantine Fault Tolerant (BFT) protocol such as PBFT~\cite{PBFT}.

This second category presents a new approach to blockchain designs and has potential for significant improvements in performance and scalability over Nakamoto consensus.
This is especially true for transaction confirmation time, since decisions in Byzantine consensus are final once committed. 
Pass and Shi~\cite{HybridConsensus} formalized the above intuition with \emph{responsiveness} notion for permissionless consensus protocols.
A protocol is \emph{responsive} if it commits transactions (possibly probabilistically) at the speed of the \emph{actual} network delay, 
without being bottlenecked by hard-coded system parameters (e.g., 10 minutes for Bitcoin). 
In the same paper, Pass and Shi proposed the hybrid consensus protocol, which uses a (slow) Nakamoto PoW chain to determine the identities of committee members, and let the committee to approve transactions responsively.

In this work, we present \name{}, 
a decentralized blockchain protocol based on reconfigurable Byzantine consensus.
We were indeed inspired by the work of Pass and Shi~\cite{HybridConsensus}. 
Yet, \name{} is conceptually very different from Bitcoin or hybrid consensus as we do not rely on Nakamoto consensus for any part of our protocol. 
Committee election and transaction processing are both done by a Byzantine consensus protocol in \name{}.    
PoW still plays a central role in \name{}, 
but we remark that use of PoW does not equate Nakamoto consensus. 
The heart of Nakamoto consensus is the idea that ``the longest PoW chain wins''.
This creates possibility of temporary ``chain forks''.
A decision in Nakamoto consensus becomes committed only if it is ``buried'' sufficiently deep in the PoW chain, 
which leads to Bitcoin's long confirmation time.

Looking from a different angle, we can think of PoW as a leader election oracle that is Sybil-proof in the permissionless setting.
But this oracle is imperfect as leader contention can still occur when multiple miners find PoWs around the same time.
Nakamoto consensus can be thought of as a probabilistic method to resolve leader contention in which miners ``vote on'' contending leaders using their mining power.
A leader (or its block) ``wins'' the contention gradually and probabilistically, until an overwhelming majority of miners ``adopt'' it by mining on top of it.
In \name{}, we also use PoW as an imperfect Sybil-proof leader election oracle. 
But instead of using Nakamoto consensus, we use a traditional Byzantine consensus protocol to resolve leader contention with certainty, rather than probabilistically.  
 

Once we have a committee, electing new members using Byzantine consensus (as opposed to Nakamoto) just seems to be the natural design once we have a closer look at the details of the protocol.
A central challenge of this framework is to reconcile the \emph{permissioned} nature of Byzantine consensus protocols and the \emph{permissionless} requirement of decentralized blockchains. 
Byzantine consensus protocols like PBFT assume a static group of committee members,
but to be permissionless and decentralized, 
committee members must change over time 
--- a step commonly referred to as \emph{reconfiguration}.
Note that the Nakamoto chain in hybrid consensus only provides agreement on the identities of the new members and when reconfiguration \emph{should} happen. 
It does not dictate when and how reconfiguration \emph{actually} happens. 
The standard reconfiguration technique requires a consensus decision from the old committee on the \emph{closing state}~\cite{VP09,Zab11,Rodrigues12,BFTSmart,Raft,BVP}.
It is then just natural to include in that consensus decision the \emph{new configuration} (i.e., the identity of the new member), at no extra cost, thereby eliminating the need for Nakamoto chains altogether. 

Not relying on Nakamoto consensus gives \name{} a few advantages over hybrid consensus.
First, the identities of committee members are exposed for a shorter duration in \name{} compared to hybrid consensus, because their identities no longer need to be buried in a Nakamoto chain.
(Admittedly, members' identities are still exposed \emph{during} their service on the committee for both \name{} and hybrid consensus.) 
Second, defending selfish mining~\cite{SelfishMining,StubbornMining} is much easier with the help of a committee (cf. Section~\ref{sec:reconfig}). 
In comparison, hybrid consensus requires quite a complex variant of Nakamoto consensus in FruitChain~\cite{FruitChain} to defend against selfish mining.
Third, analysis of our protocol is much simpler. 
In contrast, although the Bitcoin protocol is simple and elegant, its formal modeling and analysis turn out to be highly complex~\cite{GKL15,PSS16}.
FruitChain~\cite{FruitChain} and the interaction between Nakamoto and Byzantine consensus~\cite{HybridConsensus} further complicate the analysis.

\par\medskip\noindent\textbf{Contribution.}
The high-level idea of designing blockchains protocols using reconfigurable Byzantine consensus is by no means new.
PeerCensus~\cite{PeerCensus} and ByzCoin~\cite{ByzCoin} are two other works in this framework and they predate hybrid consensus and \name{}.
Comparing to those two works is tricky since their protocols were described only at a high level.
But no matter how one interprets their protocols, our paper makes the following new contributions:
\begin{itemize}
\item[--] \textbf{Detailed protocol.} 
We present full details of our protocol, and in particular, the reconfiguration step.
Reconfiguration is perhaps the step that deserves the most detailed treatment, since the protocol between two reconfiguration events is just conventional Byzantine consensus. 

\item[--] \textbf{Rigorous proof.}
We rigorously prove that \name{} achieves safety and liveness if the adversary's ratio of mining power does not exceed (roughly) 1/3. 

\item[--] \textbf{Implementation and evaluation.}
We implement \name{} and measure its performance.
Our implementation is fully Byzantine fault tolerant.
With a 0.1s network latency and 75 Mbps network bandwidth,
\name{} with committees of 1000 members can commit a consensus decision in 23.6 seconds.
\end{itemize}


\subsection{Overview of the \name{} Protocol}
\label{sec:overview}
At a high level, \name{} runs a Byzantine consensus protocol among a
set of participants, called a committee, that dynamically change over time.
The acting committee commits transactions into a linearizable log using a modified version of PBFT~\cite{PBFT}.
The log is comparable to the Bitcoin blockchain, also called a public ledger.
We say each consensus decision fills a \emph{slot} in the ledger.
A consensus decision can be either (1) a batch of transactions, (2) a reconfiguration event.
The first type of decision is analogous to a block in Bitcoin.
The second type records the membership change in the committee.

We denote the $i$-th member in chronological order as $M_i$, and the committee size is $n=3f+1$.
Then, the $i$-th committee $C_i = (M_i, M_{i+1}, M_{i+2}, \cdots, M_{i+n-1})$.
The first committee $C_1$ is known as the Genesis committee, and its $n$ members are hard coded.
After that, each new member is elected onto the committee one at a time, by the acting committee using \emph{reconfiguration consensus decisions}. 

At any time, one committee member serves as the leader.
It proposes a batch of transactions into the next empty slot $s$ in the ledger.
Other members validate proposed transactions (no double spend, etc.) and commit them in two phases. 
A Byzantine leader cannot violate safety but may prevent progress.
If members detect lack of progress, they elect the next leader in the round robin order using standard techniques from PBFT view change.
The next leader needs to stop previous leaders, learn the status from $2f+1$ members and possibly redo consensus for previous slots, before continuing on to future empty slots.

The more interesting part of protocol is reconfiguration.
To be elected onto the committee, a miner needs to present a PoW, i.e., a solution to a moderately hard computational puzzle.
Upon seeing this PoW, the current committee tries to reach a special \emph{reconfiguration consensus decision} that commits 
(1) the identity of the new member, i.e., a public key associated with the PoW, and 
(2) the closing state before the reconfiguration.
Once this reconfiguration consensus decision is committed, the system transitions into the new configuration $C_{i+1}$.
Starting from the next slot, PoW finder becomes the newest committee member $M_{i+n}$, and the oldest committee member $M_i$ loses its committee membership.

Here, a crucial design choice is: \emph{who should be the leader that drives reconfiguration?}
The first idea that comes to mind is to keep relying on the round robin leaders. 
But this approach presents a challenge on the analysis.
The current leader may be Byzantine and refuse to reconfigure when it should. 
By doing so, it tries to buy time for other adversarial miners to find competing PoWs, so that it can nominate a Byzantine new member instead.
Although we can conceive mechanisms to replace this leader, subsequent leaders may also be Byzantine.
The probability of taking in an honest new member now depends on the pattern of consecutive Byzantine leaders on the current committee.
This means we will not be able to apply the Chernoff inequality to bound the number of Byzantine members on a committee.
It is unclear whether this is merely a mathematical hurdle or the adversary really has some way to increase its representation on the committee gradually and to go above $n/3$ eventually.

Therefore, we will switch to \emph{external leaders} for reconfiguration.
In particular, the successful miner would act as the leader $L$ for reconfiguration and try to elect itself onto the committee.
When $L$'s reconfiguration proposal is committed (becomes a consensus decision), reconfiguration is finished, 
and the system starts processing transactions under the new committee with $L$ being the leader.
Note that at this point $L$ becomes a committee member, so the system has seamlessly transitioned back to internal leaders.  

By giving all external miners opportunities to become leaders, we introduce a new leader contention problem. 
It is possible that before $L$ can finish reconfiguration, another miner $L'$ also finds a PoW.
Moreover, there may be concurrent internal leaders who are still trying to propose transactions.
\name{} resolves this type of contention through a Paxos-style leader election~\cite{Paxos} with ranks. 
Only a higher ranked leader can interrupt lower ranked ones.
To ensure safety, the higher ranked leader may have to honor the proposals from lower ranked leaders by re-proposing them, if there exists one.

Now the key challenge is to figure out how leaders (internal or external) should be ranked in our protocol.
One idea is to rank external leaders by the output hash of their PoW, and stipulate that external leaders are higher ranked than all internal leaders in that configuration.
This approach, however, allows the adversary to prevent progress once in a while.
If a Byzantine miner submits a ``high ranked'' PoW but does not drive reconfiguration, the system temporarily stalls until some honest miner finds an even higher ranked PoW.
Note that no transactions can be approved in the meantime because internal leaders are all lower ranked than the Byzantine PoW finder. 
Our solution to this problem is to ``expire'' stalling external leaders, and give the leader role back to current committee members, so that the committee can resume processing transactions under internal leaders until the next external leader emerges. 
Crucially, during this process, we must enforce a total order among all leaders, internal or external, to ensure safety.
The details are presented in Section~\ref{sec:protocol_honest}.

\subsection{Overview of the Model and Proof}
We adopt the network model of Pass et al.~\cite{PSS16}: a bounded message delay of $\Delta$ that is known apriori to all participants.
Note that this is the standard synchronous network model in distributed computing, though Pass et al.~\cite{PSS16} call it ``asynchronous networks''.
What they really meant was that \emph{Bitcoin does not behave like a conventional synchronous protocol.}
In a conventional synchronous protocol, participants move forward in synchronized rounds of duration $\Delta$, 
but Bitcoin does nothing of this sort and has no clear notion of rounds.  
Our adoption of a synchronous network model may raise a few immediate questions, which we discuss below.

\par\medskip\noindent\textbf{Is the bounded message delay assumption realistic for the Internet?}
How realistic this assumption is depends a lot on the parameter $\Delta$.
With a conservative estimate, say $\Delta = $ 5 seconds, and Byzantine fault tolerance, the assumption may be believable.
Fault tolerance also helps here.
Participants experiencing slow networks and unable to deliver messages within $\Delta$ are considered Byzantine.
So the assumption we require in \name{} is that adversarial participants and ``slow'' participants collectively control no more than (roughly) 1/3 of the mining power.

More importantly, the bounded network delay assumption seems necessary for all PoW-based protocols. 
Otherwise, imagine that whenever an honest member finds a PoW, its messages are arbitrarily delayed due to asynchrony. 
Then, it is as if that the adversary controls 100\% of the mining power.
In this case, the adversary can create forks of arbitrary lengths in Bitcoin or completely take over the committee in BFT-based blockchains~\cite{PeerCensus,ByzCoin,HybridConsensus}.
Pass et al. formalized the above intuition and showed that Bitcoin's security fundamentally relies on the bounded message delay~\cite{PSS16}. 
The larger $\Delta$ is, the smaller percentage of Byzantine participants Bitcoin can tolerate; 
if $\Delta$ is unbounded (an asynchronous network), then even an adversary with minuscule mining power can break Bitcoin.

\par\medskip\noindent\textbf{Why do we use PBFT if the network is synchronous?}
It is well known that there exist protocols that can tolerate $f<n/2$ Byzantine faults in a synchronous network~\cite{KK06,XFT,SyncBC}.
By requiring $f<n/3$, PBFT preserves safety under asynchrony. 
But this seems unnecessary given that we assume synchrony.
We adopt PBFT (with modifications) because it is an established protocol that provides responsiveness. 
Most protocols that tolerate $f<n/2$ are synchronous and not responsive, because they advance in rounds of duration $\Delta$. 
If the actual latency of the network is better than $\Delta$ --- 
either because the $\Delta$ estimate is too pessimistic, the network is faster in the common case, or the network speed has improved as technology advances ---  
a responsive protocol would offer better performance than a synchronous one.
In other words, we opt to use an asynchronous protocol in a synchronous setting, essentially abandoning the strength of the bounded message delay guarantee, in order to run as fast as the network permits.
But like all PoW-based blockchains, we rely on synchrony for safety in the worse case. 
We remark that a recent protocol called XFT~\cite{XFT} seems to be responsive in its steady state while tolerating $f<n/2$ Byzantine faults. 
An interesting future direction is to explore whether we can combine our reconfiguration techniques and XFT to get the best of both worlds.

\par\medskip\noindent\textbf{Proof outline.}
If each committee in \name{} has no more than $f<n/3$ Byzantine members, then safety mostly follows from traditional Byzantine consensus. 
To guarantee $f<n/3$, we require that reconfiguration is mostly ``fair'', i.e., the probability that a newly elected member is honest (Byzantine) is roughly proportional to the mining power of honest (Byzantine) miners.
Additionally, with external leaders driving reconfiguration, we can show independence between reconfiguration events, and can then apply the Chernoff inequality to bound the probability of having $n/3$ Byzantine members on a committee.
Thus, the key step of the proof is to show that the Byzantine participants cannot distort their chance of joining the committee by too much.
Here, the adversary's main advantage is the network delay for honest participants, which essentially buys extra time for Byzantine miners to find PoWs.
But since the network delay is bounded, the Byzantine miners' advantage in finding PoWs can be also bounded.

\section{Related Work}
\label{sec:related-work}

\begin{table*}[tb]
\centering
\caption{
	\textbf{Comparison of consensus protocols in \name{} and related designs.} 
	Bitcoin and \name{} combine committee election and transactions into a single step.
	Other protocols decouple the two, which we reflect with two separate columns in the table. 
}
\label{tab:consensus_compare}
\begin{tabular}{c | c c c c}
  	\toprule
	\multirow{2}{*}{Design}		& 	\multicolumn{2}{c}{Consensus for}	& \multirow{2}{*}{Responsive} & Defend\\
						& committee election 		& transactions	& & selfish mining	\\
  	\midrule
	Bitcoin~\cite{Bitcoin} 	&\multicolumn{2}{c}{Nakamoto} 			& No & No\\
	Bitcoin-NG~\cite{Bitcoin-NG} 	& Nakamoto 		& Nakamoto		& No & No	\\
	PeerCensus~\cite{PeerCensus} 	& Byzantine$^*$ 		& Byzantine		& Yes & No	\\
	ByzCoin~\cite{ByzCoin} 			& Nakamoto/Byzantine		& Byzantine		& Yes & Yes	\\
	Hybrid consensus~\cite{HybridConsensus} 	& Nakamoto 	& Byzantine		& Yes & Yes	\\
	\hline
	\textbf{\name{}} 		& \textbf{Byzantine} 	& \textbf{Byzantine}	& \textbf{Yes} & \textbf{Yes} 	\\
  	\bottomrule
\end{tabular}
\end{table*}

\par\medskip\noindent\textbf{A comparison between \name{} and related designs.}
Table~\ref{tab:consensus_compare} compares \name{} and related designs.
To start, we have Bitcoin that pioneered the direction of permissionless consensus (though not using the term).
Bitcoin uses Nakamoto consensus, i.e., PoW and the longest-chain-win rule, and does not separate leader election and transaction processing.
There have been numerous follow-up works in the Nakamoto framework.
Some ``altcoins'' simply increases the block creation rate without other major changes to the Bitcoin protocol.
A few proposals relax the ``chain'' design and instead utilize other graph structures to allow for faster block creation rate~\cite{GHOST,LSZ15,SPECTRE}.
Some proposals~\cite{Bitcoin_lightning} aim to improve throughput with off-chain transactions.

The key idea of Bitcoin-NG is to decouple transactions and leader election~\cite{Bitcoin-NG}.
A miner is elected as a leader if it mines a (key) block in a PoW chain.
This miner/leader is then responsible for signing transactions in small batches (microblocks) until a new leader emerges.
Since (key) blocks in the PoW chain are small, Bitcoin-NG reduces the likelihood of forks.
One can think of each leader in Bitcoin-NG to be a single-member committee.
Since a single leader cannot be trusted, any transactions it approves still have to be buried in a Nakamoto PoW chain.
Thus, Bitcoin-NG does not improve transaction confirmation time.

To our knowledge, PeerCensus~\cite{PeerCensus} is the first work to envision a committee that approves transactions using PBFT.
It makes the observation that Byzantine consensus enables fast transaction confirmation.
The description and pseudocode are very high level.
Under our best interpretation, it seems that the committee is responsible for reconfiguring itself using Byzantine consensus, but no detail is given on reconfiguration.
ByzCoin~\cite{ByzCoin} employs multi-signatures to improve the scalability of PBFT. 
The reconfiguration algorithm in ByzCoin seems to be left open with multiple options.
The description in the paper~\cite{ByzCoin} is quite terse, and can be interpreted in multiple ways.
A later blog post~\cite{ByzCoinPost2} suggests adopting hybrid consensus~\cite{HybridConsensus}.
If ByzCoin goes this path, then our comparison to hybrid consensus applies to ByzCoin as well.
Through private communication, we learned that ByzCoin designers also had in mind a PBFT-style reconfiguration protocol, which seems to involve extra steps not described in the paper.
We will have to wait to see a detailed published reconfiguration protocol to compare with ByzCoin properly.
Pass and Shi proposed hybrid consensus~\cite{HybridConsensus}, which we have compared to in detail in Section~\ref{sec:intro}.

\par\medskip\noindent\textbf{Blockchain analysis.}
While the Bitcoin protocol is quite simple and elegant, proving its security rigorously turns out to be highly nontrivial. 
The original Bitcoin paper~\cite{Bitcoin} and a few early works~\cite{SelfishMining,GHOST} considered specific attacks.
Garay et al.~\cite{GKL15} provided the first comprehensive security analysis for Bitcoin, by modeling Bitcoin as a synchronous protocol running in a synchronous network.  
Pass et al.~\cite{PSS16} refined the analysis after observing that Bitcoin, while relying on a synchronous assumption, does not behave like a conventional synchronous protocol (cf. Section~\ref{sec:overview}).

\par\medskip\noindent\textbf{Byzantine Consensus.}
Consensus is a classic problem in distributed computing.
There are a few variants of the consensus problem under Byzantine faults. 
A theoretical variant is Byzantine agreement~\cite{LSP82} in which a fixed set of participants, each with an input value, try to agree on the same value. 
A more practical variant, which is also more suitable for blockchains, is BFT state machine replication (SMR).
In this variant, a fixed set of participants try to agree on a sequence of values that may come from any participant or even external sources.
Most BFT SMR protocols follow the PBFT~\cite{PBFT} framework. 


\section{Model}
\label{sec:model}

\par\medskip\noindent\textbf{Network.}
We consider a \emph{permissionless} setting in which participants use their public keys as pseudonyms and can join or leave the system at any time.
Participants are connected to each other in a peer-to-peer network with a small diameter.
As mentioned, we adopt a synchronous network model: whenever a participant sends a message to another participant, the message is guaranteed to reach the recipient within $\Delta$ time.
For convenience, we define $\Delta$ to be the end-to-end message delay bound.
If there are multiple hops from the sender to the recipient, $\Delta$ is the time upper bound for traveling all those hops.
Our protocol uses computational puzzles, i.e., PoW.
Similar to Bitcoin, the difficulty of the puzzle is periodically adjusted so that the expected time to find a PoW stays at $\puzzleTime$. 
$\puzzleTime$ should be set significantly larger than $\Delta$ according to our analysis in Section~\ref{sec:safety}.

\par\medskip\noindent\textbf{Types of participants.}
Participants in the system are either \emph{honest} or \emph{Byzantine}.
Honest participants always follow the prescribed protocol, and are capable of delivering a message to other honest participants within $\Delta$ time. 
Byzantine participants are controlled by a single adversary and can thus coordinate their actions.
At any time, Byzantine participants collectively control no more than $\rho$ fraction of the total computation power.
We assume the $n$ members in the Genesis committee $C_1$ are known to all participants, 
and that the number of Byzantine participants on $C_1$ is less than $n/3$. 

\par\medskip\noindent\textbf{Delayed adaptive adversary.}
We assume it takes time for the adversary to corrupt a honest participant.
It captures the idea that it takes time to bribe a miner or infect a clean host.
Most committee-based designs~\cite{PeerCensus,ByzCoin,HybridConsensus} require this assumption.
Otherwise, an adversary can easily violate safety by corrupting the entire committee, which may be small compared to the whole miner population.
This is formalized as a delayed adaptive adversary by Pass and Shi~\cite{HybridConsensus}. 
Specifically, we assume that even if the adversary starts corrupting a member as soon as it emerges with a PoW, by the time of corruption, the member would have already left the committee.
Whether this assumption holds in practice remains to be examined.
Algorand~\cite{Algorand} introduced techniques to hide the identities of the committee members, and can thus tolerate instant corruption.
Its core technique, however, seems to be tied to proof of stake, which has its own challenges such as the nothing-at-stake problem.

\section{The Protocol}
\label{sec:protocol_honest}

\subsection{Overview}

\name{} has a rolling committee running a Byzantine consensus protocol to commit (batches of) transactions and reconfiguration events into a ledger.
We say each committed decision fills a \emph{slot} in a linearizable ledger.
The protocol should guarantee all honest members commit the same value for each slot (safety) and keeps moving to future slots (liveness).   
To provide safety and liveness, the number of Byzantine members on each committee must be less than $f$ ($n=3f+1$), which will be proved in Section~\ref{sec:safety}. 

We first describe the protocol outside reconfiguration, i.e., when no miner has found a PoW for the current puzzle.
This part further consists of two sub-protocols: 
\emph{steady state} under a stable leader (Section~\ref{sec:steady}), and \emph{view change} to replace the leader (Section~\ref{sec:vc}).
Members monitor the current leader,
and if the leader does not make progress after some time, members support the next leader.
The new leader has to learn what values have been proposed, and possibly re-propose those values.
This part of the protocol is very similar to PBFT~\cite{PBFT} except that the views in our protocol have more structures to support external leaders for \emph{reconfiguration} (Section~\ref{sec:reconfig}).

Each configuration can have multiple \epochs{}, and each \epoch{} can have many views.
We use consecutive integers to number configurations, \epochs{} within a configuration, and views within a \epoch{}.
The \epoch{} number is the number of PoWs (honest) committee members have seen so far in the current configuration.
Intuitively, when a new PoW is found, the lifespan of the previous PoW ends.
As explained in Section~\ref{sec:overview}, our protocol can switch back and forth between internal and external leaders and we must enforce a total order among all leaders.
To this end, we will associate each leader with a configuration-\epoch-view tuple $(\CEV)$, and rank leaders by $c$ first, and then $e$, and then $v$.
Each $(c,e,v)$ tuple also defines a steady state, in which the leader of that view, denoted as $L(\CEV)$, proposes values for members to commit.
The view tuple and the leader role are managed as follows.
\begin{itemize}
\item Upon learning that a reconfiguration event is committed into the ledger, a member $M$ increases its configuration number $c$ by 1, and resets $e$ and $v$ to 0. 
$L(c, 0, 0)$ is the newly added member by the previous reconfiguration consensus decision. 
\item Upon receiving a new PoW for the current configuration, a member $M$ increases its \epoch{} number $e$ by 1, and resets $v$ to 0.
$M$ now supports the finder of the new PoW as $L(c, e, 0)$ ($e\geq1$).
\item Upon detecting that the current leader is not making progress (timed out), a member $M$ increases its view number $v$ by 1.
$L(\CEV)$ ($v\geq1$) is the $l$-th member on the current committee, by the order they join the committee, where $l = \OracleH(c, e) + v \mod n$, and $\OracleH$ is a random oracle (hash function).
\end{itemize}
To summarize, the $0$-th leader of a \epoch{} is the newly added member (for $e=0$) or the new PoW finder ($e\geq1$).
After that, leaders in that \epoch{} are selected in a round robin manner, with a pre-defined starting point.
We chose a pseudorandom starting point using $\OracleH$ but other choices also work.

Now let us consider whether honest members agree on the identities of the leaders.
For a view $(\CEV)$ with $e=0$ or $v\geq1$, the leader identity is derived from a deterministic function on the tuple and the current committee member identities (i.e., previous reconfiguration decisions).
As long as the protocol ensures safety so far, members agree on the identities of these leaders.
For views of the form $(c,e,0)$ with $e\geq1$, members may not agree on the leader's identity if they receive multiple PoWs out of order.
This is not a problem for safety because it is similar to having an equivocating leader and a BFT protocol can tolerate leader equivocation.
However, contending leaders in a view may prevent progress.
Luckily, members will time out and move on to views $(\CEV)$ with $v\geq1$.
Unique leaders in those views will ensure liveness.

Crucially, our protocol forbids pipelining to simplify reconfiguration: 
a member $M$ sends messages for slot $s$ only if it has committed all slots up to $s-1$.
Let $\Members(\CEVS)$ be the set of members currently in view $(\CEV)$ working on slot $s$.
Each honest member $M$ can only belong to one such set at any point in time.

We use digital signatures even under a stable leader, i.e., we do not use PBFT's MAC optimization~\cite{PBFT}.
The MAC optimization is not effective for a cryptocurrency since there are already hundreds of digital signatures to verify in each block/slot.
We use $\sig{x}_M$ to denote a message $x$ that carries a signature from member $M$.
A message can be signed in layers, e.g., $\sig{\sig{x}_M}_{M'}$.
When the context is clear, we omit the signer and simply write $\sig{x}$ or $\sig{\sig{x}}$.

When we say a member ``broadcasts'' a message, we mean the member sends the message to all current committee members including itself.

\subsection{Steady State}
\label{sec:steady}

\begin{itemize}
\item (\textbf{Propose})
The leader $L$ picks a batch of valid transactions $\TXs$ and then broadcasts $\TXs$ and $\sig{\Propose, \CEVS, h}_L$ where $h$ is the hash digest of $\TXs$. 

After receiving $\TXs$ and $\sig{\Propose, \CEVS, h}_L$, a member $M \in \Members(\CEVS)$ checks
\begin{itemize}
\item $L=L(\CEV)$ and $L$ has not sent a different proposal,
\item $s$ is a \emph{fresh} slot, i.e., no value could have been committed in slot $s$ (details in Section~\ref{sec:reconfig}),
\item $\TXs$ is a set of valid transactions whose digest is $h$.
\end{itemize}

\item (\textbf{Prepare})
If all the above checks pass, $M$ broadcasts $\sig{\Prepare, \CEVS, h}$.

After receiving $2f+1$ matching $\Prepare$ messages, a member $M \in \Members(\CEVS)$ \emph{accepts} the proposal (represented by its digest $h$), and concatenates the $2f+1$ matching $\Prepare$ messages into an \emph{accept certificate} $\AcceptCert$. 
 
\item (\textbf{Commit})
Upon accepting $h$, $M$ broadcasts $\sig{\Commit, \CEVS, h}$.

After receiving $2f+1$ matching $\Commit$ messages, a member $M \in \Members(\CEVS)$ commits $\TXs$ into slot $s$, and concatenates the $2f+1$ matching $\Commit$ messages into a \emph{commit certificate} $\CommitCert$.
\end{itemize}

The above is standard PBFT.
We now add an extra propagation step to the steady state.

\begin{itemize}
\item (\textbf{Notify})
Upon committing $h$, $M$ sends $\sig{\sig{\Notify, \CEVS, h}, \CommitCert}_M$ to all other members to notify them about the decision.
$M$ also starts propagating this decision on the peer-to-peer network to miners, users and merchants, etc. 
$M$ then moves to slot $s+1$.

Upon receiving a $\Notify$ message like the above, a member commits $h$, sends and propagates its own $\Notify$ message if it has not already done so, and then moves to slot $s+1$. 
\end{itemize}

\subsection{View Change}	\label{sec:vc}

A Byzantine leader cannot violate safety, but it can prevent progress by simply not proposing.
Thus, in PBFT, every honest member $M$ monitors the current leader, and if no new slot is committed after some time, 
$M$ abandons the current leader and supports the next leader.
Since we have assumed a worst-case message delay of $\Delta$, it is natural to set the timeout based on $\Delta$.
The concrete timeout values ($4\Delta$ and $8\Delta$) will be justified in the proof of Theorem~\ref{thm:safelive}.

\begin{itemize}
\item (\textbf{View-change}) 
Whenever a member $M$ moves to a new slot $s$ in a steady state, it starts a timer $T$.
If $T$ reaches $4\Delta$ and $M$ still has not committed slot $s$, then $M$ abandons the current leader and broadcasts $\sig{\Blame, \CEV}$.

Upon receiving $2f+1$ matching $\Blame$ messages for $(\CEV)$, if a member $M$ is not in a view higher than $(\CEV)$, it forwards the $2f+1$ $\Blame$ messages to the new leader $L'=L(c,e,v+1)$.
After that, if $M$ does not receive a $\NewView$ message from $L'$ within $2\Delta$, then $M$ abandons $L'$ and broadcasts $\sig{\Blame,c,e,v+1}$.

\item (\textbf{New-view})
Upon receiving $2f+1$ matching $\Blame$ messages for $(\CEV)$, the new leader $L'=L(c,e,v+1)$ concatenates them into a view-change certificate $\BlameSum$, broadcasts $\sig{\NewView,c,e,v+1,\BlameSum}_{L'}$, and enters view $(c,e,v+1)$.

Upon receiving a $\sig{\NewView,\CEV,\BlameSum}_{L'}$ message, 
if a member $M$ is not in a view higher than $(\CEV)$, it enters view $(\CEV)$ and starts a timer $T$.
If $T$ reaches $8\Delta$ and still no new slot is committed, then $M$ abandons $L'$ and broadcasts $\sig{\Blame, \CEV}$. 

\item (\textbf{Status})
Upon entering a new view $(\CEV)$, $M$ sends $\sig{\sig{\Status,\CEV, s-1, h, s, h'}, \CommitCert, \AcceptCert}$ to the new leader $L'=L(\CEV)$.  
In the above message, $h$ is the value committed in slot $s-1$ and $\CommitCert$ is the corresponding commit certificate; 
$h'$ is the value accepted by $M$ in slot $s$ and $\AcceptCert$ is the corresponding accept certificate ($h'=\AcceptCert=\bot$ if $M$ has not accepted any value for slot $s$).
We call the inner part of the message (i.e., excluding $\CommitCert,\AcceptCert$) its \emph{header}.
 
Upon receiving $2f+1$ $\Status$, $L'$ concatenates the $2f+1$ $\Status$ headers to form a $\Status$ certificate $\StatusSum$.
$L'$ then picks a $\Status$ message that reports the highest last-committed slot $s^*$; if there is a tie, $L'$ picks the one that reports the highest ranked accepted value in slot $s^*+1$ .
Let the two certificates in this message be $\CommitCert^*$ and $\AcceptCert^*$ ($\AcceptCert^*$ might be $\bot$). 

\item (\textbf{Re-propose}) 
The new leader $L'$ broadcasts $\sig{\Repropose, \CEV, s^*+1, h', \StatusSum, \CommitCert^*, \AcceptCert^*}_{L'}$.
In the above message, $s^*$ should be the highest last-committed slot reported in $\StatusSum$.
$h'$ should match the value in $\AcceptCert^*$ if $\AcceptCert^* \neq \bot$;
If $\AcceptCert^* = \bot$, then $L'$ can choose $h'$ freely.
\footnote{
Strictly speaking, in this case, $L'$ is proposing a new value rather than re-proposing.
But we keep the message name as $\Repropose$ for convenience.
A similar situation exists in reconfiguration.
}

The $\Repropose$ message serves two purposes.
First, $\CommitCert^*$ allows everyone to commit slot $s^*$.
Second, it re-proposes $h'$ for slot $s^*+1$, and carries a proof $(\StatusSum, \AcceptCert^*)$ that $h'$ is safe for slot $s^*+1$.
The $\Repropose$ message is invalid if any of the aforementioned conditions is violated: $s^*$ is not the highest committed slot, $\CommitCert$ is not for slot $s^*$, $\AcceptCert$ is not for the highest ranked accepted value for $s^*+1$, or $h'$ is not the value certified by $\AcceptCert$.

Upon receiving a valid $\Repropose$ message, a member $M$ commits slot $s^*$ if it has not already; 
$M$ then executes the prepare/commit/notify steps as in the steady state for slot $s^*+1$, and marks all slots $>s^*+1$ \emph{fresh} for view $(\CEV)$.
At this point, the system transitions into a new steady state.
\end{itemize}

A practical implementation of our protocol can piggyback $\Status$ messages on $\Blame$ messages and $\Repropose$ messages on $\NewView$ messages to save two steps.
We choose to separate them in the above description because the reconfiguration sub-protocol can reuse the $\Status$ step.

\subsection{Reconfiguration}
\label{sec:reconfig}

\par\medskip\noindent\textbf{PoW to prevent Sybil attacks.}
In order to join the committee, a new member must show a PoW, i.e., a solution to a computational puzzle based on a random oracle $\mathcal{H}$.  
The details of the puzzle will be discussed later.
For now, it is only important to know that each configuration defines a new puzzle.
Let the puzzle for configuration $c$ be $\Puzzle(c)$.
Solutions to this puzzle are in the form of $(\PuzzleSol)$.
\symb{pk} is a miner's public key.
A miner keeps trying different values of \symb{nonce} until it finds a valid solution such that 
$\mathcal{H}(\Puzzle(c), \PuzzleSol)$ is smaller than some threshold. 
The threshold is also called the \emph{difficulty} and is periodically adjusted to keep the expected reconfiguration interval at $\puzzleTime$. 

Upon finding a PoW, the miner tries to join the committee by driving a reconfiguration consensus as an external leader.
To do so, it first broadcasts the PoW to committee members.

\begin{itemize}
\item (\textbf{New-\epoch}) 
Upon finding a PoW, the miner broadcasts it to the committee members.

Upon receiving a new valid PoW for the current configuration (that $M$ has not seen before), 
$M$ forwards the PoW to other members, enters view $(c,e+1,0)$, sets the PoW finder as its 0-th leader $L'=L(c,e+1,0)$ of the new \epoch{}, and starts a timer $T$.
If $T$ reaches $8\Delta$ and still no new slot is committed, then $M$ abandons $L'$ and broadcasts $\sig{\Blame,c,e+1,0}$.

\item (\textbf{Status})
Upon entering a new \epoch{}, $M$ sends $L'$ $\sig{\sig{\Status,\CEZ, s-1, h, s, h'}, \CommitCert, \AcceptCert}$, where $s,h,\CommitCert,\AcceptCert$ are defined the same way as in the view-change sub-protocol.

Upon receiving $2f+1$ $\Status$, $L'$ constructs $s^*$, $\StatusSum$, $\CommitCert^*$ and $\AcceptCert^*$ as in view-change.

\item (\textbf{Re-propose and Reconfigure})
Let $h^*$ be the committed value in the highest committed slot $s^*$ among $\StatusSum$.
Let $h'$ be the highest ranked accepted value (could be $\bot$) into slot $s^*+1$ among $\StatusSum$.
Note that $h^*$ and $h'$ are certified by $\StatusSum, \CommitCert^*$ and $\AcceptCert^*$.
Now depending on whether $h^*$ and $h'$ are transactions or reconfiguration events, the new external leader $L'$ takes different actions.

\begin{itemize}
\item If $h^*$ is a reconfiguration event to configuration $c+1$, then $L'$ simply broadcasts $\CommitCert$ and terminates.
Terminating means $L'$ gives up its endeavor to join the committee (it is too late and some other leader has already finished reconfiguration) and starts working on $\Puzzle(c+1)$.

\item If $h^*$ is a batch of transactions, and $h'$ is a reconfiguration event to configuration $c+1$, then $L'$ broadcasts $\sig{\Repropose, \CEZ, s^*+1, h', \StatusSum, \CommitCert^*, \AcceptCert^*}_{L'}$ and then terminates.

\item If $h^*$ is a batch of transactions, and $h'=\bot$, then $L'$ tries to drive the reconfiguration consensus into slot $s^*+1$ by broadcasting
$\sig{\Repropose, \CEZ, s^*+1, h, \StatusSum, \CommitCert^*, \AcceptCert^*}_{L'}$ where $h$ is a reconfiguration event that lets $L'$ join the committee.
If this proposal becomes committed, then starting from slot $s^*+2$, the system reconfigures to the next configuration and $L'$ joins the committee replacing the oldest member. 

\item If $h^*$ and $h'$ are both batches of transactions, then $L$' first re-proposes $h'$ for slot $s^*+1$ by broadcasting $\sig{\Repropose, \CEZ, s^*+1, h', \StatusSum, \CommitCert^*, \AcceptCert^*}_{L'}$.
After that, $L'$ tries to drive a reconfiguration consensus into slot $s^*+2$ by broadcasting
$\sig{\Propose, \CEZ, s^*+2, h}_{L'}$ where $h$ is a reconfiguration event that lets $L'$ join the committee.
If this proposal becomes committed, then starting from slot $s^*+3$, the system reconfigures to the next configuration and $L'$ joins the committee replacing the oldest member. 
\end{itemize}

\end{itemize}

In the latter three cases, members react to the $\Repropose$ message in the same way as in view-change: check its validity and execute prepare/commit/notify.
Note that the reconfiguration proposal in the last case is a steady state $\Propose$ message instead of a $\Repropose$ message, and it does not need to carry certificates as proofs.
This is because after members commit the re-proposed value into slot $s^*+1$, the system enters a special steady state under leader $L'$. 
Slot $s^*+2$ will be marked as a \emph{fresh} slot in this steady state.
This steady state is special in the sense that, 
if $L'$ is not faulty and is not interrupted by another PoW finder,
members will experience a configuration change, but the leader remains unchanged as we have defined $L(c+1,0,0)$ to be the newly added member. 
After reconfiguration, it becomes a normal steady state in view $(c+1,0,0)$.

\par\medskip\noindent\textbf{The puzzle and defense to selfish mining.}
We now discuss what exactly the puzzle is.
A natural idea is to make the puzzle be the previous reconfiguration decision. 
However, this allows \emph{withholding attacks} similar to selfish mining~\cite{SelfishMining}.
To favor the adversary, we assume that if the adversary and an honest miner both find PoWs, the adversary wins the contention and joins the committee.
Then, when the adversary finds a PoW first, its best strategy is to temporarily withhold it. 
At this point, the adversary already knows the next puzzle, which is its own reconfiguration proposal.
Meanwhile, honest miners are wasting their mining power on the current puzzle, not knowing that they are doomed to lose to the adversary when they eventually find a PoW.
This gives the adversary a head start on the next puzzle.
Its chance of taking the next seat on the committee will hence be much higher than its fair share of mining power.
The authors of ByzCoin suggested an idea to thwart the above attack: let the puzzle include $2f+1$ committee member signatures.
This way, a PoW finder cannot determine the puzzle on its own and thus gains nothing from withholding its PoW.
For concreteness, we let $\Puzzle(c+1)$ be \emph{any} set of $f+1$ $\Notify$ headers for the last reconfiguration decision $\sig{\Notify, \CEVS, h}$. This ensures the following:

\begin{lemma}
The adversary learns a puzzle at most $2\Delta$ time earlier than honest miners.
\label{lemma:puzzle}
\end{lemma} 
\begin{proof}
For the adversary to learn the puzzle, at least one honest committee member must broadcast its $\Notify$ message, which will cause all other honest members to broadcast their $\Notify$ messages within $\Delta$ time.
After another $\Delta$ time, all honest miners receive enough $\Notify$ messages to learn the puzzle.
\end{proof}

\subsection{Safety and Liveness}
\label{sec:safety}

We first prove safety and liveness \emph{assuming} each committee has no more than $f<n/3$ Byzantine members.
The proof of the following theorem (given in the appendix)  mostly follows PBFT. 
Essentially, the biggest change we made to PBFT at an algorithmic level is from a round robin leader schedule to a potential interleaving between internal and external leaders.
But crucially, we still have a total order among leaders to ensure safety and liveness. 

\begin{theorem}
The protocol achieves safety and liveness if each committee has no more than $f<n/3$ Byzantine members.
\label{thm:safelive}
\end{theorem} 

Next, we show $f<n/3$ indeed holds.
Let $\rho$ be the adversary's ratio of mining power.
Honest miners thus collective control $1-\rho$ of the total network mining power.
Ideally, we would hope that reconfiguration is a ``fair game'', i.e., the adversary taking the next committee seat is with probability $\rho$, and an honest miner takes the next seat with probability $1-\rho$.
This is indeed simply assumed to be the case in PeerCensus~\cite{PeerCensus} and ByzCoin~\cite{ByzCoin}.
However, we show that due to network latency,
it is as if that the adversary has an \emph{effective mining power} $\rho' = 1 - (1-\rho)e^{-(2\rho+8)\Delta/\puzzleTime} > \rho$.
(Recall that $\puzzleTime$ is the expected time for the entire network to find a PoW.)
We need $\puzzleTime \gg \Delta$ so that $\rho'$ is not too much larger than $\rho$.

\begin{theorem}
Assuming $f<n/3$ holds for $C_1$, then $f<n/3$ holds for each subsequent committee except for a probability exponentially small in $n$ if $\rho' = 1 - (1-\rho)e^{-(2\rho+8)\Delta/\puzzleTime} < 1/3$.
\label{thm:committee_quality}
\end{theorem} 

\begin{proof}
The increase in the adversary's effective mining power comes from three sources.
First, as Lemma~\ref{lemma:puzzle} shows, the adversary may learn a puzzle up to $2\Delta$ earlier than honest miners.
Second, if the adversary finds a PoW while an honest external leader is reconfiguring, which can take up to $8\Delta$ time, the adversary can interrupt the honest leader and win the reconfiguration race.
Lastly, two honest leaders may interrupt each other if they find PoWs within $8\Delta$ time apart, while the adversary-controlled miners can coordinate their actions.
Therefore, an honest miner wins a reconfiguration if all three events below happen:
\begin{itemize}
\item (event $X$) the adversary finds no PoW in its $2\Delta$ head start,
\item (event $Y$) some honest miner finds a PoW earlier than the adversary given $X$, and 
\item (event $Z$) no miner, honest or adversarial, finds a PoW in the next $8\Delta$ time given $Y$.
\end{itemize}

PoW mining is a \emph{memory-less} process modeled by a Poisson distribution: 
the probability for a miner to find $k$ PoWs in a period is $p(k,\lambda)= \frac{\lambda^k e^{-\lambda}} {k!}$,
where $\lambda$ is the expected number of PoWs to be found in this period.
Recall that the expected time for all miners combined to find a PoW is $\puzzleTime$.
In the $2\Delta$ head start, the adversary expects to find $\lambda_X = 2\Delta\rho/\puzzleTime$ PoWs.
Thus, $\Pr(X) = p(0, \lambda_X) = e^{-2\Delta\rho/\puzzleTime}$. 
Similarly, $\Pr(Z) = e^{-8\Delta/\puzzleTime}$.
Event $Y$ is a ``fair race'' between the adversary and honest miners, so $\Pr(Y) = 1-\rho$.
The memory-less nature of PoW also means that the above events $X, Y$ and $Z$ are independent.
Thus, the probability that an honest miner wins a reconfiguration is 
$$\Pr(X \cap Y \cap Z) = \Pr(X) \Pr(Y) \Pr(Z) = (1-\rho)e^{-(2\rho+8)\Delta/\puzzleTime} := 1 - \rho'.$$ 
We define the above probability to be $1-\rho'$, which can be thought of as the honest miners' effective mining power once we take message delays into account.
$\rho'$ is then the adversary's effective mining power.
This matches our intuition: as $\puzzleTime \rightarrow \infty$, the adversary's advantage from message delay decreases, 
and $(1-\rho') \rightarrow (1-\rho)$. 

Conditioned on all committees up to $C_i$ having no more than $n/3$ Byzantine members, the adversary wins each reconfiguration race with probability at most $\rho'$, independent of the results of other reconfiguration races.
We can then use the Chernoff inequality to bound the probability of Byzantine members exceeding $n/3$ on any committee up to $C_{i+1}$.
Let $Q$ denote the number of Byzantine members on a committee. 
We have $E(Q) := \mu = \rho'n$.
$\Pr(Q \geq (1+\delta)\mu) \leq e^{\frac{-\delta \mu \log(1+\delta) }{2}}$ due to the Chernoff bound.
Select $\delta = \frac{1}{3\rho'} - 1$,
we have $\Pr(Q \geq n/3) \leq e^{\frac{n(1-3\rho')\log (3\rho')}{6}}$. 
If $\rho'<1/3$, then $(1-3\rho') \log(3\rho') < 0$, and the above probability is exponentially small in $n$ as required.
\end{proof}

\par\medskip\noindent\textbf{Concrete committee size.}
We calculate the required committee size under typical values.
We can calculate $\rho'$ from $\rho$ given the value of $\Delta/\puzzleTime$.
If $\Delta$ is 5 seconds and $\puzzleTime$ is 10 minutes, then $\Delta/\puzzleTime = 1/120$.
Then, for a desired security level, a simple calculation of binomial distribution yields the required committee size $n$.
The results are listed in Table~\ref{tab:committee_size}.
A security parameter $k$ means there is a $2^{-k}$ probability that, after a reconfiguration, the adversary controls more than $n/3$ seats on the committee.
$k=25$ or $k=30$ should be sufficient in practice.
Assuming we reconfigure every 10 minutes, 
$2^{25}$ reconfigurations take more than 600 years and $2^{30}$ reconfigurations take more than 20,000 years.
With $\rho'=25\%$, the required committee size is around 1000.

It is worth noting that as $\rho'$ approaches 33\%, the required committee size increases rapidly and the committee-based approach becomes impractical.
This gap between theory and practice is expected.
Similarly, Bitcoin, in theory, can tolerate an adversarial with up to 50\% mining power. 
But the recommended ``six confirmation'' is calculated assuming a 10\% adversary and a rather low security level $k \approx 10$. 
If the adversary really controls close to 50\% mining power, then thousands of confirmations would be required for each block.

\begin{table}[tb]
\centering
\caption{The required committee size under different adversarial miner ratios $\rho$ and desired security levels $k$, assuming $\Delta/\puzzleTime = 1/120.$}
\label{tab:committee_size}
\begin{tabular}{c | c | c  c  c  c  c}
\toprule
\multicolumn{2}{c}{$k$} & 20 & 25 & 30 & 35 & 40 \\
\midrule
$\rho=14\%$ & $\rho'=20\%$ & 232 & 298 & 367  & 439 & 508 \\
$\rho=20\%$ & $\rho'=25\%$ & 649 & 841 & 1036 & 1231 & 1423 \\
$\rho=23\%$ & $\rho'=28\%$ & 1657 & 2149 & 2644 & 3142 & 3580 \\
$\rho=25\%$ & $\rho'=30\%$ & 4366 & 5650 & 6949 & 8248 & 9256 \\
\bottomrule
\end{tabular}
\end{table}

\subsection{Towards Pipelining}
One limitation of the above \name{} protocol is that it forbids pipelining, i.e., members are not allowed to work on multiple slots in parallel. 
The interaction between pipelining and reconfiguration is non-trivial, 
and to the best of our knowledge, has not been worked out in BFT protocols.
In particular, when a slot is being worked on, the exact configuration (i.e., identities of committee members) for that slot would not be known if previous slots have not been committed. 
In this subsection, we briefly describe how we plan to support pipelining. 
We hope to present and implement a detailed pipelined protocol in a future version.

The key idea is to let each proposal carry a digest of its ``predecessor proposal''. 
Since the predecessor proposal also carries the digest of its own predecessor proposal, each proposal is tied to its entire chain of ``ancestor proposals''. 
A non-leader member $M$ sends messages for a proposal for slot $s$ only if $M$ has accepted its predecessor proposal in slot $s-1$. 
When a member accepts a proposal, it also accepts all its ancestors proposals; 
likewise, when a member commits a proposal, it also commits all its ancestors proposals. 

During a view-change or a reconfiguration, a member reports in its $\Status$ message the highest ranked proposal chain that it has accepted (with an accept certificate).
The new leader $L'$ then re-proposes the highest ranked proposal chain it hears from members.
Note that in either step, we prioritize the highest ranked proposal chain, not the longest proposal chain.
If the highest ranked proposal chain contains a reconfiguration event in its tail (highest numbered slot), then $L'$ re-proposes this proposal chain and terminates.
Otherwise, $L'$ re-proposes the proposal chain and then start making its own proposals.
If $L'$ is a PoW finder, then its proposal will be a reconfiguration proposal.
If its proposal goes through, $L'$ joins the committee and the system enters a new configuration.
The first leader in the new configuration is the newly added member $L'$, and the system transitions back to a steady state with $L'$ now serving as an internal leader.

\section{Implementation and Evaluation}	\label{sec:eval}

\par\medskip\noindent\textbf{Implementation details.}
We implement the non-pipelined version of the \name{} consensus 
protocol in Python and measure its performance.
We implement our 6-phase PBFT-style protocol in which committee members agree on a proposal digest from a successful miner.
For digital signatures, we use ECDSA with the prime192v2 curve and the Charm crypto library~\cite{CharmCrypto}. 
We test \name{} on 16 m4.16xlarge machines on Amazon EC2. 
The CPUs are Intel Xeon E5-2686 v4 @ 2.3 GHz.
Each machine has 64 cores, giving a total of 1024 cores.
We use 1 CPU core for each committee member, testing committee size up to 1000.
To emulate Internet latency and bandwidth, 
we add a delay of 0.1 second (unless otherwise stated) to each message and throttle the bandwidth of each committee member to the values in Figure~\ref{fig:results}. 

\par\medskip\noindent\textbf{Experimental results.}
The main result we report is the time it takes for committee members to commit a reconfiguration consensus decision (committing a decision in the steady state takes strictly less time).
We report the consensus time from the leader's perspective, i.e., the time between when the leader sends its PoW and when the leader receives the first $\sym{Notify}$ message, which is the earliest point at which the leader knows its proposal is committed.
Figure~\ref{fig:results} presents the time to reach a reconfiguration consensus decision (averaged over 10 experiments) under different committee sizes $n$ and bandwidth limits.
As shown in the zoomed-in part, with a small committee of 100 members, 
a decision takes 1 to 3 seconds depending on the throughput.
The bottleneck here is network latency.
Recall that there are 8 sequential messages from the leader's perspective, each adding 0.1s latency.
With larger committees of 400 and 1000 members, network bandwidth becomes the bottleneck.
As expected, consensus time is roughly quadratic in committee size and
inversely proportional to network bandwidth. 
With moderate network bandwidth, consensus time is manageable even with 1000 committee members: 48.3s for 35 Mbps and 23.6s for 75 Mbps.
To study the effect of network latency, we rerun some experiments with a 0.5s latency (results not shown).
In each experiment, the consensus time increases by about 3s,
which is as expected given the 8 sequential messages and a 0.4s latency increase for each.

\begin{figure}[bt]      
  \centering 
  \includegraphics[width=0.56\columnwidth]{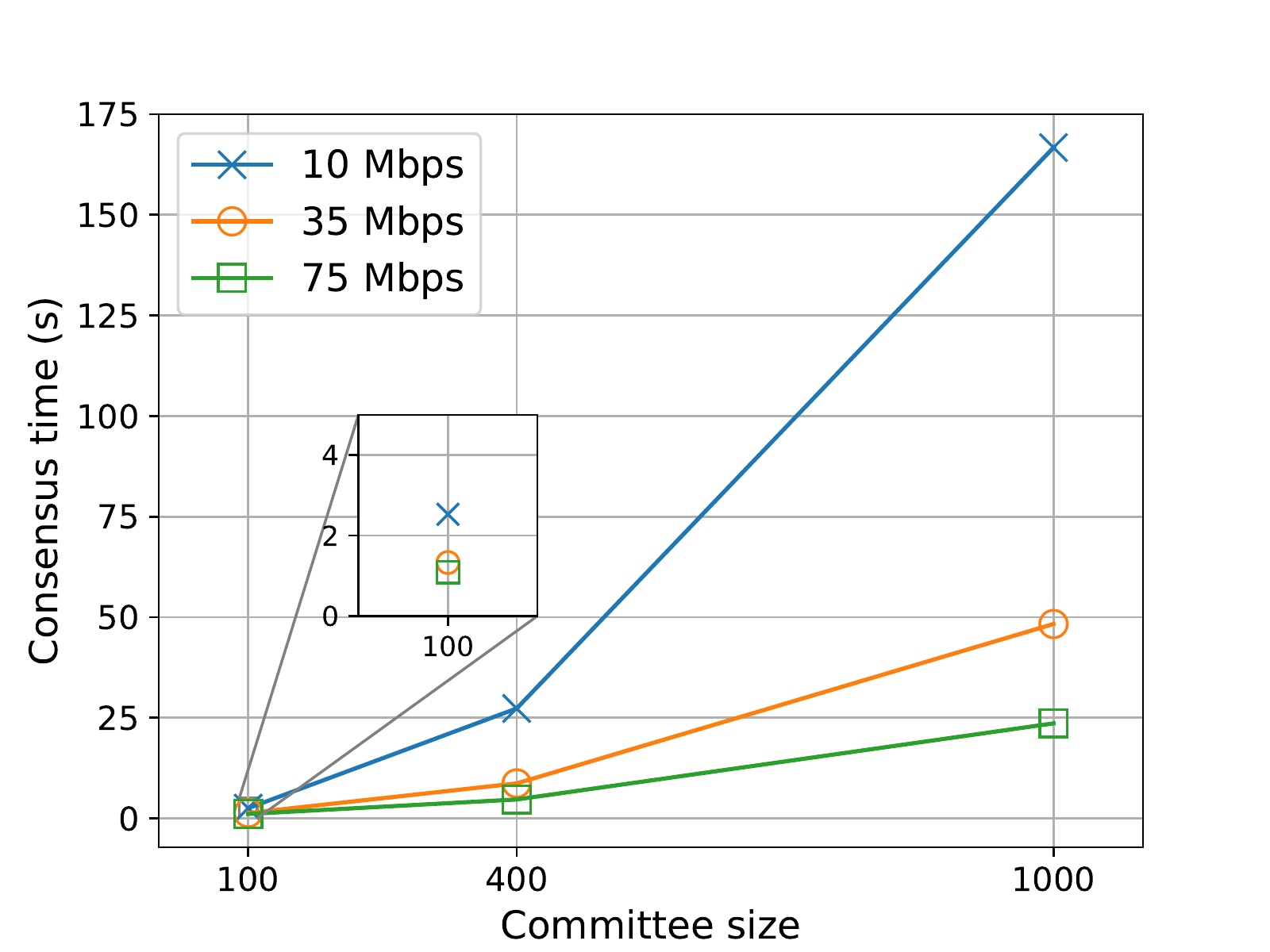}
  \caption{ 
  	Time to achieve reconfiguration consensus with different committee sizes and network bandwidth.}
  \label{fig:results}
\end{figure}

\section{Conclusion and Future Work}
\label{sec:conclusion}
This work presents \name{}, a blockchain protocol based on PoW-augmented reconfigurable Byzantine consensus.
We have presented a detailed protocol, rigorously proved safety and liveness under the (seemingly necessary) bounded message delay model, and provided a prototype implementation and evaluation results. 
Besides the pipelined protocol, interesting future directions include designing and rigorously analyzing incentives for \name{}, and extending \name{} to tolerate an adversary with up to 50\% mining power.


\section*{Acknowledgement}
We thank Eleftherios Kokoris-Kogias for clarifying the ByzCoin protocol and other helpful discussions.
We thank the anonymous reviewers for their valuable suggestions. 
This work is funded in part by NSF award \#1518765, a Google Ph.D. Fellowship award, and an Azrieli Foundation scholarship. 

{
\small
\bibliographystyle{plain}
\bibliography{blockchain}

\begin{thebibliography}{10}

\bibitem{BVP}
Ittai Abraham and Dahlia Malkhi.
\newblock {BVP}: Byzantine vertical paxos, 2016.

\bibitem{CharmCrypto}
Joseph~A Akinyele, Christina Garman, Ian Miers, Matthew~W Pagano, Michael
  Rushanan, Matthew Green, and Aviel~D Rubin.
\newblock Charm: a framework for rapidly prototyping cryptosystems.
\newblock {\em Journal of Cryptographic Engineering}, 3(2):111--128, 2013.

\bibitem{BFTSmart}
Alysson Bessani, Jo{\~a}o Sousa, and Eduardo~EP Alchieri.
\newblock State machine replication for the masses with bft-smart.
\newblock In {\em Dependable Systems and Networks (DSN), 2014 44th Annual
  IEEE/IFIP International Conference on}, pages 355--362. IEEE, 2014.

\bibitem{BLS01}
Dan Boneh, Ben Lynn, and Hovav Shacham.
\newblock Short signatures from the weil pairing.
\newblock pages 514--532. Springer, 2001.

\bibitem{PBFT}
Miguel Castro and Barbara Liskov.
\newblock Practical byzantine fault tolerance.
\newblock In {\em Proceedings of the third symposium on Operating systems
  design and implementation}, pages 173--186. USENIX Association, 1999.

\bibitem{Bitcoin_tx_cnt}
CoinDesk.
\newblock Average number of transactions per block, accessed Aug, 2016.
\newblock https://www.coindesk.com/data/bitcoin-number-transactions-per-block/.

\bibitem{PeerCensus}
Christian Decker, Jochen Seidel, and Roger Wattenhofer.
\newblock Bitcoin meets strong consistency.
\newblock In {\em Proceedings of the 17th International Conference on
  Distributed Computing and Networking}, page~13. ACM, 2016.

\bibitem{Bitcoin-NG}
Ittay Eyal, Adem~Efe Gencer, Emin~G{\"u}n Sirer, and Robbert Van~Renesse.
\newblock {Bitcoin-NG}: A scalable blockchain protocol.
\newblock In {\em 13th USENIX Symposium on Networked Systems Design and
  Implementation (NSDI 16)}, pages 45--59, 2016.

\bibitem{SelfishMining}
Ittay Eyal and Emin~G{\"u}n Sirer.
\newblock Majority is not enough: Bitcoin mining is vulnerable.
\newblock In {\em International Conference on Financial Cryptography and Data
  Security}, pages 436--454. Springer, 2014.

\bibitem{ByzCoinPost2}
Bryan Ford.
\newblock Untangling mining incentives in bitcoin and byzcoin, 2016.
\newblock http://bford.github.io/2016/10/25/mining/.

\bibitem{GKL15}
Juan Garay, Aggelos Kiayias, and Nikos Leonardos.
\newblock The {Bitcoin} backbone protocol: Analysis and applications.
\newblock In {\em Annual International Conference on the Theory and
  Applications of Cryptographic Techniques}, pages 281--310. Springer, 2015.

\bibitem{Zab11}
Flavio~P Junqueira, Benjamin~C Reed, and Marco Serafini.
\newblock Zab: High-performance broadcast for primary-backup systems.
\newblock In {\em Dependable Systems \& Networks (DSN), 2011 IEEE/IFIP 41st
  International Conference on}, pages 245--256. IEEE, 2011.

\bibitem{KK06}
Jonathan Katz and Chiu-Yuen Koo.
\newblock On expected constant-round protocols for byzantine agreement.
\newblock {\em J. Comput. Syst. Sci.}, 75(2):91--112, 2009.

\bibitem{ByzCoin}
Eleftherios~Kokoris Kogias, Philipp Jovanovic, Nicolas Gailly, Ismail Khoffi,
  Linus Gasser, and Bryan Ford.
\newblock Enhancing {Bitcoin} security and performance with strong consistency
  via collective signing.
\newblock In {\em 25th USENIX Security Symposium}, pages 279--296. USENIX
  Association, 2016.

\bibitem{Paxos}
Leslie Lamport.
\newblock The part-time parliament.
\newblock {\em ACM Transactions on Computer Systems (TOCS)}, 16(2):133--169,
  1998.

\bibitem{VP09}
Leslie Lamport, Dahlia Malkhi, and Lidong Zhou.
\newblock Vertical paxos and primary-backup replication.
\newblock In {\em Proceedings of the 28th ACM symposium on Principles of
  distributed computing}, pages 312--313. ACM, 2009.

\bibitem{LSP82}
Leslie Lamport, Robert Shostak, and Marshall Pease.
\newblock The byzantine generals problem.
\newblock {\em ACM Transactions on Programming Languages and Systems (TOPLAS)},
  4(3):382--401, 1982.

\bibitem{LSZ15}
Yoad Lewenberg, Yonatan Sompolinsky, and Aviv Zohar.
\newblock Inclusive block chain protocols.
\newblock In {\em International Conference on Financial Cryptography and Data
  Security}, pages 528--547. Springer, 2015.

\bibitem{XFT}
Shengyun Liu, Christian Cachin, Vivien Qu{\'e}ma, and Marko Vukolic.
\newblock {XFT}: practical fault tolerance beyond crashes.
\newblock In {\em 12th USENIX Symposium on Operating Systems Design and
  Implementation}, pages 485--500. USENIX Association, 2016.

\bibitem{Elastico16}
Loi Luu, Viswesh Narayanan, Chaodong Zheng, Kunal Baweja, Seth Gilbert, and
  Prateek Saxena.
\newblock A secure sharding protocol for open blockchains.
\newblock In {\em Proceedings of the 2016 ACM SIGSAC Conference on Computer and
  Communications Security}, pages 17--30. ACM, 2016.

\bibitem{Algorand}
Silvio Micali.
\newblock Algorand: The efficient and democratic ledger.
\newblock arXiv:1607.01341, 2016.

\bibitem{Bitcoin}
Satoshi Nakamoto.
\newblock Bitcoin: A peer-to-peer electronic cash system, 2008.

\bibitem{StubbornMining}
Kartik Nayak, Srijan Kumar, Andrew Miller, and Elaine Shi.
\newblock Stubborn mining: Generalizing selfish mining and combining with an
  eclipse attack.
\newblock In {\em 2016 IEEE European Symposium on Security and Privacy
  (EuroS\&P)}, pages 305--320. IEEE, 2016.

\bibitem{Raft}
Diego Ongaro and John~K Ousterhout.
\newblock In search of an understandable consensus algorithm.
\newblock In {\em USENIX Annual Technical Conference}, pages 305--319, 2014.

\bibitem{PSS16}
Rafael Pass, Lior Seeman, and Abhi Shelat.
\newblock Analysis of the blockchain protocol in asynchronous networks.
\newblock In {\em Advances in Cryptology -- EUROCRYPT 2017}, pages 643--673.
  Springer International Publishing, 2017.

\bibitem{HybridConsensus}
Rafael Pass and Elaine Shi.
\newblock Hybrid consensus: Efficient consensus in the permissionless model.
\newblock 2016.
\newblock Cryptology ePrint Archive, Report 2016/917.

\bibitem{FruitChain}
Rafael Pass and Elaine Shi.
\newblock Fruitchains: A fair blockchain.
\newblock In {\em Proceedings of the ACM Symposium on Principles of Distributed
  Computing}, pages 315--324. ACM, 2017.

\bibitem{Bitcoin_lightning}
Joseph Poon and Thaddeus Dryja.
\newblock The {Bitcoin} lightning network: Scalable off-chain instant payments.
\newblock Technical Report. https://lightning.network, 2015.

\bibitem{SyncBC}
Ling Ren, Kartik Nayak, Ittai Abraham, and Srinivas Devadas.
\newblock Practical synchronous byzantine consensus.
\newblock Cryptology ePrint Archive, Report 2017/307, 2017.
\newblock \url{http://eprint.iacr.org/2017/307}.

\bibitem{Rodrigues12}
Rodrigo Rodrigues, Barbara Liskov, Kathryn Chen, Moses Liskov, and David
  Schultz.
\newblock Automatic reconfiguration for large-scale reliable storage systems.
\newblock {\em IEEE Transactions on Dependable and Secure Computing},
  9(2):145--158, 2012.

\bibitem{SPECTRE}
Yonatan Sompolinsky, Yoad Lewenberg, and Aviv Zohar.
\newblock {SPECTRE}: A fast and scalable cryptocurrency protocol, 2016.

\bibitem{GHOST}
Yonatan Sompolinsky and Aviv Zohar.
\newblock Secure high-rate transaction processing in bitcoin.
\newblock In {\em International Conference on Financial Cryptography and Data
  Security}, pages 507--527. Springer, 2015.

\end{thebibliography}
}

\newpage
\appendix

\section{Proofs}

\begin{proof}[Proof of Theorem~\ref{thm:safelive}]
We first consider safety.
Let $M$ be the member to commit slot $s$ in the lowest ranked view.
Say $M$ commits value $h$ and it does so in view $(\CEV)$.
We need to show that no other value $h' \neq h$ can be committed into slot $s$ in that view and all future views.
In fact, we will show by induction that there cannot exist an accept certificate $\AcceptCert'$ for $h'$ in that view and all future views, which means no honest member would have accepted $h'$, let alone committing it.
Since $M$ committed in view $(\CEV)$, there must be $2f+1$ members that have accepted $h$ in that view, and $2f+1$ members that have sent $\sig{\Prepare, \CEV, h}$.
For the base case, suppose for contradiction that $\AcceptCert'$ exists for $h'$ in view $(\CEV)$. 
Then, $2f+1$ members must send $\sig{\Prepare, \CEV, h'}$.
This means at least one honest member has sent $\Prepare$ messages for two different proposals in view $(\CEV)$, which is a contradiction.
For the inductive step, suppose that no accept certificate exists for $h'$ from view $(\CEV)$ up to (excluding) view $(c',e',v')$.
Since $M$ has committed slot $s$, then at least $2f+1$ members must have committed slot $s-1$ (they send $\Prepare$ for slot $s$ only after committing slot $s-1$).
So in the status certificate $\StatusSum$ of a $\Repropose$ message for view $(c',e',v')$, the largest last-committed slot $s^* \geq s-1$.
Therefore, slot $s$ is either already committed to $h$, or has to be re-proposed.
If slot $s$ needs to be re-proposed, which means $s^*=s-1$, then $\StatusSum$ must contain at least one $\Status$ header reporting that $h$ has been accepted into slot $s$ with a rank no lower than $(c,e,v)$.
Due to the inductive hypothesis, no accept certificate can exist for $h'$ with a rank equal to or higher than $(c,e,v)$.
$h$ is, therefore, the unique highest ranked accepted value for slot $s$, and it is the only value that can be legally re-proposed in view $(c',e',v')$.
Hence, no accept certificate for $h'$ can exist in view $(c',e',v')$, completing the proof.

Next, we consider liveness.
According to the protocol, if a Byzantine leader does not commit any slot for too long, it will be replaced shortly.
But here is an another liveness attack that a Byzantine leader can perform in the cryptocurrency setting.
The Byzantine leader can simply construct transactions that transfer funds between its own accounts and keep proposing/committing these transactions.
It is impossible for honest members to detect such an attack.
Fortunately, we have shown that honest miners win at least 2/3 of the reconfiguration races, and after a reconfiguration, the newly added member becomes the leader.
So 2/3 of the time, an honest leader is in control to provide liveness.
It remains to check that an honest leader will not be replaced until the next \epoch{}.
It is sufficient to show that an honest leader will never be accused by an honest member.
This is guaranteed by our timeout values.
First, observe that the $\Notify$ step ensures two honest members commit a slot at most $\Delta$ time apart.
In the steady state, each member gives the leader $4\Delta$ time to commit each slot: one $\Delta$ to account for the possibility that other members come to this slot (i.e., finish the previous slot) $\Delta$ time later, and three $\Delta$ for the three phases in the steady state. 
At the beginning of a new view or a new \epoch{}, each member gives the leader $8\Delta$ time, because the $\Status$ and $\Repropose$ steps take extra $4\Delta$ time.
The last case to consider is when a member abandons a future leader because it does not send $\NewView$ when it should.
Before accusing a future leader, a member allows $2\Delta$ time after forwarding it a view-change certificate, which is sufficient for an honest leader to broadcast $\NewView$.

\end{proof}

\par\medskip\noindent\textbf{Remark.}
Garay et al.~\cite{GKL15} laid out three desired properties for blockchains: common prefix, chain quality, and chain growth. 
We note that these three properties are more applicable to Bitcoin-style blockchains that do not satisfy traditional safety and liveness.
Common prefix is strictly weaker than safety.
Chain growth is the Nakamoto consensus counterpart of liveness.
If chain quality is interpreted as ``the ratio of honest committee members'' in the BFT-based approach, then we have proved it in Theorem~\ref{thm:committee_quality}.
If it is instead interpreted as ``the ratio of slots committed under honest leaders'', then it is not very meaningful. 
We cannot prevent Byzantine leaders from making progress really fast, thereby decreasing the ratio of slots committed under honest leaders, but that does not hurt honest leaders' ability to commit slots at their own pace in their own views.

\section{Comparison to ByzCoin Evaluation Results}
We would like to explain an apparent contradiction between our experimental results and those of ByzCoin~\cite{ByzCoin}.
Under the same parameter settings (0.1s latency, 35 Mbps bandwidth and 100 committee members),
our implementation of PBFT takes only 1.3 seconds per consensus decision,
while ByzCoin concludes that PBFT has unacceptable performance.
As a result, ByzCoin suggests sending (and aggregating) consensus messages in a tree rooted at the leader. 
On a closer look, our results and theirs actually corroborate each other.
ByzCoin finds PBFT performance to be unacceptable only when the block is large.
With a small block size, their results confirm that PBFT is very fast and indeed takes about 1s per consensus decision. 
Therefore, the unacceptable performance they see in PBFT solely results from the leader broadcasting the entire batch of transactions $\TXs$. 
It is indeed important to reduce the burden on the leader by gossiping $\TXs$ among committee members.  
But using a tree as the communication graph is not Byzantine fault tolerant. 
A Byzantine node in the tree can make its entire subtree unreachable.
Instead, $\TXs$ should be sent through the peer-to-peer network much like in Bitcoin. 
The other major technique ByzCoin proposed, the Schnorr multi-signature, is also not Byzantine fault tolerant.
Schnorr signature has a ``commitment'' step before the actual signing step.
A Byzantine signer may participate in the commitment step but refuse to sign in the signing step.
The leader then has to initiate a new instance of Schnorr multi-signature up to $f$ times, which may be even slower than using normal signatures.
Therefore, the results reported in ByzCoin are only for the best case where there is no adversary.
These results are still relevant if we believe most of the time during the protocol, there is no adversary.  
Alternatively, we may use multi-signature schemes that do not have the commitment step and are thus Byzantine fault tolerant~\cite{BLS01}.
But these schemes are much less efficient, and it remains interesting future work to study whether they improve efficiency in practice. 

\end{document}